\documentclass[copyright,creativecommons]{eptcs}
\providecommand{\event}{GASCom 2026}

\usepackage{iftex}

\ifpdf
  \usepackage{underscore}
  \usepackage[T1]{fontenc}
\else
  \usepackage{breakurl}
\fi

\usepackage{amsmath,amssymb,amsthm}
\usepackage{booktabs}

\newtheorem{theorem}{Theorem}
\newtheorem{proposition}[theorem]{Proposition}

\newtheorem{conjecture}[theorem]{Conjecture}
\newtheorem{definition}[theorem]{Definition}

\title{Exhaustive Generation of Genus-One Knot and Link Diagrams\\
via Maps on the Torus}

\author{Alexander Omelchenko
\institute{Constructor University Bremen\\
Campus Ring 1, 28759 Bremen, Germany}
\email{aomelchenko@constructor.university}
}

\newcommand{\titlerunning}{Exhaustive Generation of Genus-One Diagrams via Maps on the Torus}
\newcommand{\authorrunning}{A. Omelchenko}
\hypersetup{
  bookmarksnumbered,
  pdftitle={Exhaustive Generation of Genus-One Diagrams via Maps on the Torus},
  pdfauthor={A. Omelchenko},
  pdfsubject={GASCom 2026},
  pdfkeywords={enumeration algorithm, canonical labelling, combinatorial enumeration, maps on surfaces, knot tabulation, Kauffman bracket}
}

\begin{document}
\maketitle

\begin{abstract}
We present an algorithmic framework for the exhaustive generation and tabulation
of knot and link diagrams on the thickened torus $T^2\times I$,
based on the theory of maps on surfaces.
Cellular $4$-regular torus projections are encoded by permutation pairs
$(\alpha,\sigma)$, and unsensed equivalence classes are enumerated completely
and without duplication via canonical representatives.
Crossing assignments, local diagram-level reductions, and the generalized
Kauffman-type bracket are formulated entirely within the same permutation model.
The pipeline is validated against published genus-one classifications for
crossing numbers $N\le 5$ and then extended to $N=6,7,8$, producing,
to our knowledge, the first complete genus-one tabulation at these crossing
numbers under the stated comparison conventions.
The resulting dataset contains more than $33\,000$ knot and link types.
Besides the tables, the computation yields proved structural facts, including
a parity statement for the $a$-span of the bracket and a sharp upper bound
$N-1$ for the number of bigon faces in a $4$-regular torus map.
It also suggests several conjectures, among them a formula for the maximum
number of straight-ahead components, the absence of equi-quadrilateral knot
projections, and a $4N$ upper bound for the genus-one bracket span.
\end{abstract}

\section{Introduction}

Tabulation of knots and links has a long tradition going back to Tait~\cite{Tait1877},
culminating in very large classical tables~\cite{HosteThistlethwaiteWeeks1998}.
Virtual knots~\cite{Kauffman1999} and links in thickened surfaces
$\Sigma_g\times I$ provide a natural extension of this program beyond~$S^3$.
A fundamental result of Kuperberg~\cite{Kuperberg2003} asserts that every virtual link
admits a unique minimal-genus surface representative.
We work throughout with this surface-based viewpoint.

For the thickened torus $T^2\times I$, explicit genus-one tables
at small crossing numbers appeared in~\cite{Omelchenko2007,Grishanov2009a,
Grishanov2009b,AkimovaMatveev2012,AkimovaMatveev2014,AkimovaMatveevTarkaev2020}.
These works typically start from abstract $4$-regular graphs, then classify
inequivalent embeddings into the surface, and only afterwards assign crossing data.
Controlling equivalence under surface homeomorphisms is the main bottleneck.

The central idea of this paper is that the embedding should be treated as the
\emph{primary} combinatorial object.
This is precisely the subject of the classical theory of \emph{maps on surfaces}.
A cellular embedding of a connected graph into a closed orientable surface is a map,
and a cellular $4$-regular knot projection is simply a $4$-regular map.
Maps admit a compact permutation encoding~\cite{LandoZvonkin2004}:
an involution $\alpha$ pairs the two darts of each edge, and a permutation $\sigma$
records the cyclic order of darts at each vertex.

This approach is related to general-purpose map and graph generators, but it is not
a direct invocation of them.  The program \texttt{plantri} is optimized for many
classes of plane maps and planar graphs~\cite{BrinkmannMcKay2007}.  More recent
work of Brinkmann gives a general generator for maps on oriented surfaces by combining
abstract graph generators with embedding and isomorphism rejection routines~\cite{Brinkmann2025}.
Our setting is more specialized: the generated objects are loopless cellular
$4$-regular torus maps satisfying knot-theoretic filters, and the subsequent crossing
assignments and state-sum invariants are computed in the same dart-permutation model.
Thus Brinkmann's framework provides a natural independent benchmark or alternative
generator, whereas the present implementation is designed as an integrated pipeline
for genus-one knot and link diagrams.

In this encoding, unsensed equivalence of projections reduces to a finite conjugacy
relation on permutation pairs, so that canonical representatives yield a
duplication-free enumeration without a separate classification of embeddings.
Moreover, the number of straight-ahead components (hence the knot/link distinction)
is read off directly from the permutations.

Building on the map-enumeration line developed for degree-$4$
structures~\cite{Krasko,KraskoOmelchenko2019PartI,
KraskoOmelchenko2019PartII,KraskoOmelchenko2021NoPlanarLoops},
we turn these map-theoretic tools into a structural engine for knot
and link tabulation on~$T^2$.
The implementation is validated against published genus-one tables for $N\le 5$
and then extended to $N=6,7,8$.
A public reference implementation and datasets are
available~\cite{OmelchenkoTorusMaps}.

\section{The Permutation Model for Maps}\label{sec:maps}

Let $\Sigma_g$ be the closed orientable surface of genus~$g$.
A \emph{map} on $\Sigma_g$ is a cellular embedding of a connected graph
into $\Sigma_g$.
Permutations act on the right: $(pq)(h)=q(p(h))$.

\begin{definition}[Permutation encoding]
A labelled orientable map is encoded by a pair $(\alpha,\sigma)$
of permutations on a dart set $H$ such that:
(i)~$\alpha$ is a fixed-point-free involution (edge pairing);
(ii)~the cycles of $\sigma$ record the cyclic order of darts at vertices;
(iii)~$\langle \alpha,\sigma\rangle$ acts transitively on $H$ (connectedness).
The face permutation is $\varphi:=\sigma\alpha$, and the genus is determined by
\[
\#\mathrm{cyc}(\sigma)-\#\mathrm{cyc}(\alpha)+\#\mathrm{cyc}(\varphi)=2-2g.
\]
\end{definition}

Two encodings define the same \emph{sensed} map if one is obtained from the other
by simultaneous conjugation; they define the same \emph{unsensed} map if,
additionally, $\sigma\mapsto\sigma^{-1}$ (orientation reversal) is allowed.

For $4$-regular maps, each $\sigma$-cycle has length~$4$, and the involution
$\sigma^2$ pairs opposite darts at each vertex.
Define $\rho:=\sigma^2\alpha$.
Then the number of \emph{straight-ahead components} of the corresponding projection
(i.e. the components of the projected curve) is
\[
  c(P)=\tfrac12\,\#\mathrm{cyc}(\rho).
\]
This detects the knot/link distinction directly from the permutation data.

Monogon faces (fixed points of~$\varphi$) are the projection-level shadow of
immediate Reidemeister~I reductions and are excluded from our enumeration.

\paragraph{Example.}
Consider the dart set $H=\{1,\dots,8\}$ with
$\sigma_0=(1\,2\,3\,4)(5\,6\,7\,8)$ ($N=2$ vertices) and
$\alpha=(1\,6)(2\,8)(3\,5)(4\,7)$.
Then $\varphi=\sigma_0\alpha=(1\,7)(2\,5)(3\,8)(4\,6)$, giving $F=4$ cycles
(each of length~$2$) and
$V-E+F=2-4+4=2\neq 0$; this is a sphere map, not a torus one.
The torus condition $F=N$ is therefore a genuine constraint on the matching~$\alpha$.

\section{Canonical Enumeration of Projections on the Torus}\label{sec:projections}

Fix $N\ge 1$ and let $H=\{1,\dots,4N\}$.
We fix a standard vertex rotation
\[
\sigma_0=(1\,2\,3\,4)(5\,6\,7\,8)\cdots(4N{-}3\ 4N{-}2\ 4N{-}1\ 4N).
\]
Every unlabelled $4$-regular map admits a labelled encoding of the form
$(\alpha,\sigma_0)$ for some fixed-point-free involution~$\alpha$.

\begin{definition}[Candidate torus projection]
A pair $(\alpha,\sigma_0)$ is a candidate projection on $T^2$ with $N$
crossings if:
(1)~$\alpha$ is a fixed-point-free involution;
(2)~$\langle\alpha,\sigma_0\rangle$ acts transitively on $H$;
(3)~$\#\mathrm{cyc}(\sigma_0\alpha)=N$ (torus condition);
(4)~$\varphi=\sigma_0\alpha$ has no fixed points (no monogons);
(5)~no loop edges (for every $h$, $h$ and $\alpha(h)$ lie in
different $\sigma_0$-cycles).
\end{definition}

\paragraph{Canonical representatives and rooted normalization.}
Fix a total order $\prec$ on labelled pairs.
The unsensed canonical representative of $(\alpha,\sigma)$ is the
$\prec$-minimal element of the union of the conjugacy orbits of $(\alpha,\sigma)$
and $(\alpha,\sigma^{-1})$ under $S_H$.
Two encodings represent the same unsensed map if and only if their canonical
representatives coincide.

In practice this minimum is computed by rooted normalization rather than by enumerating
$S_H$.  For each root dart $r$ and for each of the two orientations
$\sigma$ and $\sigma^{-1}$, we perform a deterministic traversal of the dart graph,
inspecting from a dart $h$ the neighbours in the fixed order $(\sigma(h),\alpha(h))$.
Darts are relabelled $1,2,\dots,4N$ in order of first discovery, and the resulting
labelled pair is recorded.  The canonical form is the lexicographically smallest
normal form over all $4N$ roots and both orientations.  A rooted connected map has
no non-trivial label ambiguity once the root and orientation are fixed, so this
procedure gives the same complete invariant as the abstract orbit minimum.
Its cost is $O(N)$ per root and hence $O(N^2)$ per accepted labelled encoding.

\paragraph{Enumeration algorithm and crude complexity bounds.}
The conceptual algorithm enumerates fixed-point-free involutions $\alpha$ on $H$;
for each one it tests the defining conditions, computes the canonical representative,
and stores it.  After deduplication, this yields exactly one representative per
unsensed class.
The naive worst-case number of matchings is
\[
  (4N-1)!!=\frac{(4N)!}{2^{2N}(2N)!},
\]
so the corresponding crude time bound is $O(N^2(4N-1)!!)$ and the storage needed
after deduplication is $O(N|\mathcal C_N|)$.
This bound is intentionally pessimistic.  The implementation uses a canonical
construction path with activated vertices and early rejection of partial matchings;
it does not materialize all matchings.  The lower bound is the output size:
any complete enumeration must take $\Omega(|\mathcal C_N|)$ time.

\begin{theorem}[Correctness]\label{thm:correctness}
The procedure outputs a set $\mathcal{C}_N$ that is sound (every element is a valid
candidate), complete (every unsensed class is represented), and non-duplicating
(no two elements represent the same class).
\end{theorem}

\paragraph{Primeness filters.}
A projection is called \emph{prime} in the computational sense used here if its
underlying multigraph has no $2$-edge-cut and no splitness witness based on mixed
vertices.  These are explicit combinatorial witnesses, stable under unsensed
equivalence.  They are sufficient for matching the projection-level conventions used
in the comparison tables below, but they are not presented as a complete topological
characterization of primeness in every possible convention.

\paragraph{Benchmark data.}
Table~\ref{tab:projections} records the sizes of the enumerated sets.
Table~\ref{tab:linksByComponents} shows the distribution of prime link
projections by component count.

\begin{table}[ht]
\centering
\begin{tabular}{c|ccc|cc}
\toprule
$N$ & candidates & removed & prime total & knot proj. & link proj. \\
\midrule
$3$ & $6$   & $0$    & $6$    & $2$    & $4$   \\
$4$ & $28$  & $5$    & $23$   & $10$   & $13$  \\
$5$ & $109$ & $28$   & $81$   & $34$   & $47$  \\
$6$ & $595$ & $216$  & $379$  & $170$  & $209$ \\
$7$ & $3216$ & $1421$ & $1795$ & $777$  & $1018$ \\
$8$ & $19956$ & $10141$ & $9815$ & $4308$ & $5507$ \\
\bottomrule
\end{tabular}
\caption{Unsensed $4$-regular projection counts on the torus for $N=3,\ldots,8$.
The column ``removed'' records removals by the $2$-edge-cut compositeness witness;
in this range the splitness witness removes no further candidates.}
\label{tab:projections}
\end{table}

\begin{table}[ht]
\centering
\begin{tabular}{r|ccccc}
\toprule
$N$ & $c=2$ & $c=3$ & $c=4$ & $c=5$ & $c=6$ \\
\midrule
$3$ & $3$    & $1$   & ---  & ---  & ---  \\
$4$ & $9$    & $3$   & $1$  & ---  & ---  \\
$5$ & $37$   & $9$   & $1$  & ---  & ---  \\
$6$ & $150$  & $51$  & $7$  & $1$  & ---  \\
$7$ & $775$  & $212$ & $30$ & $1$  & ---  \\
$8$ & $4030$ & $1293$ & $169$ & $14$ & $1$ \\
\bottomrule
\end{tabular}
\caption{Prime link projections by number of straight-ahead components $c(P)\ge 2$.}
\label{tab:linksByComponents}
\end{table}

\section{From Projections to Diagrams}\label{sec:diagrams}

For a fixed projection $P$ with $N$ crossings, a \emph{diagram} $D=(P,b)$ is
obtained by choosing a crossing assignment $b:\mathrm{Vert}(P)\to\{0,1\}$,
selecting which strand is overpassing at each vertex.
Thus $b$ is the over/under datum of the diagram.
This gives $2^N$ diagrams per projection, or half as many when the comparison
convention identifies the global crossing switch $b\mapsto 1-b$.

\paragraph{Bigon reduction.}
A bigon face is a $2$-cycle $(i\,j)$ of $\varphi=\sigma\alpha$.
Suppose that it is incident to two distinct vertices $u$ and $v$.
Choose local cyclic orders
$(h_0\,h_1\,h_2\,h_3)$ at $u$ and $(k_0\,k_1\,k_2\,k_3)$ at $v$,
where $b(u)=0$ means that $\{h_0,h_2\}$ is overpassing and $b(v)=0$ means
that $\{k_0,k_2\}$ is overpassing.
If $i=h_r$ and $j=k_s$, put $p_u(i)=r\pmod 2$ and $p_v(j)=s\pmod 2$.
Then the bigon supports an immediate Reidemeister~II reduction if and only if
\[
 b(u)+b(v)+p_u(i)+p_v(j)\equiv 1 \pmod 2.
\]
The reduced-diagram convention used in the comparison tables keeps precisely those
assignments for which this congruence fails for every two-vertex bigon.
The simpler rule $b(u)\ne b(v)$ is valid only in local labellings where
$p_u(i)=p_v(j)=0$.

\paragraph{State-sum invariants.}
The generalized Kauffman bracket on the torus~\cite{AkimovaMatveev2012}
is formulated entirely in the permutation model.
At a vertex, the crossing assignment $b$ selects the overpassing local strand.
An $A$- or $B$-smoothing replaces the crossing by one of the two non-crossing
pairings of the four incident darts; the product over all vertices is a smoothing
involution $\tau_{b,s}$, where the state $s$ records the $A/B$ choice at each vertex.
The \emph{state circles} are the connected components of the $2$-regular dart graph
whose edges are the $\alpha$-pairs and the $\tau_{b,s}$-pairs.  Equivalently, each
state circle gives two cycles of the permutation $\pi_{b,s}:=\alpha\tau_{b,s}$.

The torus bracket distinguishes contractible and essential state circles.
Let $C_1$ be the $\mathbb F_2$-vector space generated by the projection edges and
let $\partial_2:C_2\to C_1$ be the cellular boundary map sending each face to its
mod-$2$ boundary vector.  A state circle is contractible if and only if its edge
incidence vector lies in $\operatorname{im}\partial_2$; otherwise it is essential.
Thus no drawing in a fundamental polygon is required.
The bracket is
\begin{equation}\label{eq:bracket}
\langle D\rangle(a,x)=\sum_{s\in\{A,B\}^N}
 a^{A(s)-B(s)}(-a^2-a^{-2})^{\gamma(s)}x^{\delta(s)},
\end{equation}
where $A(s)$ and $B(s)$ are the numbers of $A$- and $B$-smoothings in the state,
and $\gamma(s)$ and $\delta(s)$ count contractible and essential state circles.
For knots, $w(D)$ denotes the usual writhe, and the normalized polynomial
\[
  X_D(a,x)=(-a)^{-3w(D)}\langle D\rangle(a,x)
\]
is invariant under all Reidemeister moves.

An efficient implementation precomputes the pair $(\gamma(t),\delta(t))$ for every
smoothing vector $t\in\{0,1\}^N$ once per projection, then reuses the data for all
crossing assignments via the bijection $s=b\oplus t$.

\paragraph{Validation.}
For $N\le 5$ we compare raw locally reduced diagram classes with the published
lists.  The comparison is summarized in Table~\ref{tab:diagramValidation}.

\begin{table}[ht]
\centering
\begin{tabular}{c|c|c}
\toprule
$N$ & knots & links \\
\midrule
$2$ & $1/1$ & $1/1$ \\
$3$ & $3/3$ & $4/4$ \\
$4$ & $18/17$ & $22/21$ \\
$5$ & $71/69$ & $99/99$ \\
\bottomrule
\end{tabular}
\caption{Raw locally reduced counts: this work / published count.}
\label{tab:diagramValidation}
\end{table}

The discrepancies are convention-dependent and are not treated as contradictions
with the published classifications.  For $N=4$ knots, the additional class in our raw
list is supported by the unique $N=4$ projection without bigon faces; it is excluded
in the published table by a global primeness convention not captured by our local
bigon rule.  For $N=4$ links, the published count $21$ is obtained from an original
list of $22$ by a global equivalence already noted in the link classification.
For $N=5$ knots, one of the two extra raw classes is equivalent to a smaller-crossing
class under the comparison conventions.  The remaining difference reflects an additional
global convention in the published knot table, such as the central-fragment convention
related to $\Omega_3$ moves, rather than a failure of the permutation enumeration itself.
For $N=5$ links the count agrees exactly with the published table when the optional
over/under participation prefilter is not imposed.

\section{Tabulation Data and Structural Observations}\label{sec:observations}

The computation is extended to $N=6,7,8$.
Table~\ref{tab:growth} records the new diagram types at each crossing number.

\begin{table}[ht]
\centering
\begin{tabular}{r|rrr|r}
\toprule
$N$ & knots & links & total & ratio to $N{-}1$ \\
\midrule
$2$ & $1$     & $1$      & $2$      & ---   \\
$3$ & $3$     & $4$      & $7$      & $3.5$ \\
$4$ & $18$    & $22$     & $40$     & $5.7$ \\
$5$ & $71$    & $99$     & $170$    & $4.2$ \\
$6$ & $378$   & $525$    & $903$    & $5.3$ \\
$7$ & $1743$  & $2909$   & $4652$   & $5.2$ \\
$8$ & $10704$ & $16752$  & $27456$  & $5.9$ \\
\bottomrule
\end{tabular}
\caption{New genus-one diagram types at each crossing number.
Through $N=8$ the tabulation contains $12\,918$ knot and $20\,312$ link types
($33\,230$ total).}
\label{tab:growth}
\end{table}

A log-linear fit over $N=3,\dots,8$ is numerically close to
$0.052\cdot 5.14^N$, but this is only a descriptive summary of the available data.
With six data points we make no asymptotic claim about the true exponential growth
constant.

\subsection{Maximum number of straight-ahead components}

For each $N$, let $c_{\max}(N)$ denote the maximum number of straight-ahead
components among prime projections.
In the computed range $3\le N\le 8$:

\begin{proposition}[Lattice lower bound]\label{prop:lattice}
For every even $N\ge 4$, let $P_N$ be the $4$-regular map on the torus
obtained from the square lattice $\mathbb{Z}^2$ by identifying under
$\Lambda_N=\langle(2,0),(0,N/2)\rangle$.
Then $P_N$ has $N$ vertices, $N$ quadrilateral faces, $N/2+2$
straight-ahead components, and no $2$-edge-cut.
\end{proposition}

\begin{conjecture}\label{conj:cmax}
For all $N\ge 3$, $c_{\max}(N)=\lfloor N/2\rfloor+2$.
In the computed range the maximizer is unique: the lattice projection for even $N$;
a projection with face-degree vector $(3,3,4^{N-3},6)$ for odd $N$.
\end{conjecture}

The proved even-$N$ lattice construction gives the lower bound.
The remaining problem is the upper bound.  A possible approach is to study the
intersection pattern of straight-ahead components as a collection of essential curves
on the torus and to show that primeness forces enough mixed vertices between them.
For odd $N$, the computed maximizers suggest a triangle-hexagon defect inserted into
the even lattice family; constructing and analyzing this family explicitly would be a
natural first step.

\begin{table}[ht]
\centering
\begin{tabular}{r|cccccc}
\toprule
$N$ & $3$ & $4$ & $5$ & $6$ & $7$ & $8$ \\
\midrule
$c_{\max}(N)$ & $3$ & $4$ & $4$ & $5$ & $5$ & $6$ \\
$\lfloor N/2\rfloor+2$ & $3$ & $4$ & $4$ & $5$ & $5$ & $6$ \\
\bottomrule
\end{tabular}
\caption{Maximum component count among prime torus projections.}
\label{tab:cmax}
\end{table}

\subsection{Equi-quadrilateral projections are always links}

A $4$-regular torus map is \emph{equi-quadrilateral} if every face has degree~$4$.
In the computed range, every prime equi-quadrilateral projection has $c(P)\ge 2$:
no equi-quadrilateral projection supports a knot diagram.

\begin{conjecture}\label{conj:equiquad}
Every prime equi-quadrilateral $4$-regular map on the torus has at least
$2$ straight-ahead components.
\end{conjecture}

This admits a purely permutation-theoretic reformulation:
if $\sigma$ and $\varphi=\sigma\alpha$ both have cycle type $(4^N)$ and
$\langle\sigma,\alpha\rangle$ acts transitively, must
$\rho=\sigma^2\alpha=\sigma\varphi$ have at least four cycles?
Equivalently, can the straight-ahead component count be one?
The computation suggests that homological restrictions on the straight-ahead cycles
may rule out this case.

\subsection{The $a$-span of the genus-one bracket}

For a Laurent polynomial $Q(a,x)$, define
\[
\operatorname{span}_a(Q)=
\max\{k:[a^k]Q\ne0\}-\min\{k:[a^k]Q\ne0\}.
\]

\begin{proposition}[Span parity]\label{prop:parity}
For any torus diagram with $N$ crossings, every monomial in
$\langle D\rangle(a,x)$ has $a$-exponent of parity~$N$.
In particular, $\operatorname{span}_a\langle D\rangle$ is even.
\end{proposition}

\begin{proof}
In the state sum~\eqref{eq:bracket}, $A(s)-B(s)=2A(s)-N\equiv N\pmod{2}$, and
$(-a^2-a^{-2})^{\gamma(s)}$ contributes only even exponents.
\end{proof}

\begin{conjecture}[$4N$ span bound]\label{conj:span}
For every genus-one knot diagram with $N$ crossings,
$\operatorname{span}_a(X_D)\le 4N$.
The same bound holds for the bracket of link diagrams in the computed range
$3\le N\le 8$.
\end{conjecture}

This is motivated by the classical Kauffman--Murasugi--Thistlethwaite span theorem,
but the present statement is only conjectural in genus one.  A possible route to a
proof would be to adapt the state-graph argument behind the classical span inequality,
while controlling cancellations between states with contractible and essential circles.
Table~\ref{tab:span} confirms the bound in the computed range.

\begin{table}[ht]
\centering
\begin{tabular}{r|cccc}
\toprule
$N$ & min span & max span & $4N$ & avg span \\
\midrule
$3$ &  $6$ & $12$ & $12$ & $9.0$ \\
$4$ &  $8$ & $16$ & $16$ & $11.9$ \\
$5$ & $10$ & $20$ & $20$ & $15.4$ \\
$6$ &  $4$ & $24$ & $24$ & $17.7$ \\
$7$ &  $4$ & $28$ & $28$ & $21.0$ \\
$8$ &  $4$ & $32$ & $32$ & $23.6$ \\
\bottomrule
\end{tabular}
\caption{Observed $a$-span of $X_D(a,x)$ for genus-one knot types.}
\label{tab:span}
\end{table}

\subsection{Further observations}

\paragraph{Purely essential knots.}
A genus-one knot is \emph{purely essential} if $X_D(a,x)$ has no $x^0$ term.
Equivalently at the level of the invariant, all purely contractible contributions
cancel or are absent.  This phenomenon, vacuous on $S^2$, accounts for a large
fraction of genus-one knot types (Table~\ref{tab:essential}).

\begin{table}[ht]
\centering
\begin{tabular}{r|rr|r}
\toprule
$N$ & purely ess. & total knots & fraction \\
\midrule
$3$ &    $2$ &     $3$ & $0.67$ \\
$4$ &   $14$ &    $18$ & $0.78$ \\
$5$ &   $34$ &    $71$ & $0.48$ \\
$6$ &  $276$ &   $378$ & $0.73$ \\
$7$ & $1053$ &  $1743$ & $0.60$ \\
$8$ & $7805$ & $10704$ & $0.73$ \\
\bottomrule
\end{tabular}
\caption{Fraction of purely essential knot types at each crossing number.}
\label{tab:essential}
\end{table}

\paragraph{Bigon prevalence.}
A $4$-regular torus map with $N$ vertices has at most $N-1$ bigon faces: since
Euler's formula gives $F=N$, not all faces can be bigons, and if $N-1$ faces are
bigons then the remaining face has degree $2N+2$.
This bound is sharp at every $N\le 8$ in our data.
The fraction of prime projections containing at least one bigon reaches $97\%$ at
$N=8$.  This trend is qualitatively consistent with general $0$--$1$ laws for submaps:
in broad map classes, an admissible fixed planar submap appears linearly often with
probability tending to one~\cite{BenderGaoRichmond1992}.  Our loopless, $4$-regular,
prime torus projections form a more constrained family, so this theorem does not
by itself prove the observed statement here; it explains why rapid appearance of
bigons should not be surprising.

\section{Outlook}\label{sec:outlook}

The permutation framework applies without modification to any $\Sigma_g$.
On the torus, the contractible/essential dichotomy for state circles reduces to
an $\mathbb{F}_2$ boundary test, and a single variable~$x$ suffices for the bracket.
For $g\ge 2$, separating essential circles can be null-homologous, requiring
integral homology classes and mapping-class-group canonicalization; this
is the main obstacle to scaling the pipeline beyond genus one.

\bibliographystyle{eptcs}

\bibliography{references}

\begin{thebibliography}{10}
\providecommand{\bibitemdeclare}[2]{}
\providecommand{\surnamestart}{}
\providecommand{\surnameend}{}
\providecommand{\urlprefix}{Available at }
\providecommand{\url}[1]{\texttt{#1}}
\providecommand{\href}[2]{\texttt{#2}}
\providecommand{\urlalt}[2]{\href{#1}{#2}}
\providecommand{\doi}[1]{doi:\urlalt{https://doi.org/#1}{#1}}
\providecommand{\eprint}[1]{arXiv:\urlalt{https://arxiv.org/abs/#1}{#1}}
\providecommand{\bibinfo}[2]{#2}

\bibitemdeclare{article}{AkimovaMatveevTarkaev2020}
\bibitem{AkimovaMatveevTarkaev2020}
\bibinfo{author}{S.~V.~Matveev \surnamestart A.~A.~Akimova\surnameend} \&
  \bibinfo{author}{V.~V. \surnamestart Tarkaev\surnameend}
  (\bibinfo{year}{2020}): \emph{\bibinfo{title}{Classification of prime links
  in the thickened torus having crossing number 5}}.
\newblock {\slshape \bibinfo{journal}{J. of Knot Theory and Its Ramifications}}
  \bibinfo{volume}{29}(\bibinfo{number}{03}), p. \bibinfo{pages}{2050012},
  \doi{10.1142/S0218216520500121}.

\bibitemdeclare{article}{AkimovaMatveev2012}
\bibitem{AkimovaMatveev2012}
\bibinfo{author}{A.~A. \surnamestart Akimova\surnameend} \&
  \bibinfo{author}{S.~\surnamestart Matveev\surnameend} (\bibinfo{year}{2012}):
  \emph{\bibinfo{title}{Classification of knots in T x I with at most 4
  crossings}}.
\newblock {\slshape \bibinfo{journal}{arXiv: Geometric Topology}}.
\newblock \urlprefix\url{https://api.semanticscholar.org/CorpusID:119621553}.

\bibitemdeclare{article}{AkimovaMatveev2014}
\bibitem{AkimovaMatveev2014}
\bibinfo{author}{A.~A. \surnamestart Akimova\surnameend} \&
  \bibinfo{author}{S.~V. \surnamestart Matveev\surnameend}
  (\bibinfo{year}{2014}): \emph{\bibinfo{title}{Classification of genus 1
  virtual knots having at most five classical crossings}}.
\newblock {\slshape \bibinfo{journal}{J. of Knot Theory and Its Ramifications}}
  \bibinfo{volume}{23}(\bibinfo{number}{06}), p. \bibinfo{pages}{1450031},
  \doi{10.1142/S021821651450031X}.

\bibitemdeclare{article}{Brinkmann2025}
\bibitem{Brinkmann2025}
\bibinfo{author}{G.~\surnamestart Brinkmann\surnameend} (\bibinfo{year}{2024}):
  \emph{\bibinfo{title}{Generating Maps on Oriented Surfaces Using the
  Homomorphism Principle}}.
\newblock {\slshape \bibinfo{journal}{Discrete \& Computational Geometry}},
  \doi{10.1007/s00454-025-00734-5}.

\bibitemdeclare{article}{BrinkmannMcKay2007}
\bibitem{BrinkmannMcKay2007}
\bibinfo{author}{G.~\surnamestart Brinkmann\surnameend} \&
  \bibinfo{author}{B.~D. \surnamestart McKay\surnameend}
  (\bibinfo{year}{2007}): \emph{\bibinfo{title}{Fast generation of planar
  graphs}}.
\newblock {\slshape \bibinfo{journal}{Match-communications in Mathematical and
  in Computer Chemistry}} \bibinfo{volume}{58}, pp. \bibinfo{pages}{323--357}.
\newblock \urlprefix\url{https://api.semanticscholar.org/CorpusID:116425311}.

\bibitemdeclare{article}{BenderGaoRichmond1992}
\bibitem{BenderGaoRichmond1992}
\bibinfo{author}{Z.-C.~Gao \surnamestart E.~A.~Bender\surnameend} \&
  \bibinfo{author}{L.~B. \surnamestart Richmond\surnameend}
  (\bibinfo{year}{1992}): \emph{\bibinfo{title}{Submaps of maps. I. General 0-1
  laws}}.
\newblock {\slshape \bibinfo{journal}{J. Comb. Theory {B}}}
  \bibinfo{volume}{55}(\bibinfo{number}{1}), pp. \bibinfo{pages}{104--117},
  \doi{10.1016/0095-8956(92)90034-U}.

\bibitemdeclare{article}{HosteThistlethwaiteWeeks1998}
\bibitem{HosteThistlethwaiteWeeks1998}
\bibinfo{author}{M.~Thistlethwaite \surnamestart J.~Hoste\surnameend} \&
  \bibinfo{author}{J.~\surnamestart Weeks\surnameend} (\bibinfo{year}{1998}):
  \emph{\bibinfo{title}{The first 1,701,936 knots}}.
\newblock {\slshape \bibinfo{journal}{The Mathematical Intelligencer}}
  \bibinfo{volume}{20}, pp. \bibinfo{pages}{33--48}, \doi{10.1007/BF03025227}.

\bibitemdeclare{article}{Kauffman1999}
\bibitem{Kauffman1999}
\bibinfo{author}{L.~H. \surnamestart Kauffman\surnameend}
  (\bibinfo{year}{1999}): \emph{\bibinfo{title}{Virtual Knot Theory}}.
\newblock {\slshape \bibinfo{journal}{Eur. J. Comb.}}
  \bibinfo{volume}{20}(\bibinfo{number}{7}), pp. \bibinfo{pages}{663--691},
  \doi{10.1006/EUJC.1999.0314}.

\bibitemdeclare{article}{Krasko}
\bibitem{Krasko}
\bibinfo{author}{E.~\surnamestart Krasko\surnameend} \&
  \bibinfo{author}{A.~\surnamestart Omelchenko\surnameend}
  (\bibinfo{year}{2017}): \emph{\bibinfo{title}{Enumeration of 4-regular
  one-face maps}}.
\newblock {\slshape \bibinfo{journal}{European J. of Combinatorics}}
  \bibinfo{volume}{62}, pp. \bibinfo{pages}{167--177},
  \doi{10.1016/j.ejc.2016.12.004}.

\bibitemdeclare{article}{KraskoOmelchenko2019PartI}
\bibitem{KraskoOmelchenko2019PartI}
\bibinfo{author}{E.~\surnamestart Krasko\surnameend} \&
  \bibinfo{author}{A.~\surnamestart Omelchenko\surnameend}
  (\bibinfo{year}{2019}): \emph{\bibinfo{title}{Enumeration of r-regular maps
  on the torus. Part~I:}}.
\newblock {\slshape \bibinfo{journal}{Discrete Math.}}
  \bibinfo{volume}{342}(\bibinfo{number}{2}), pp. \bibinfo{pages}{584--599},
  \doi{10.1016/j.disc.2018.07.013}.

\bibitemdeclare{article}{KraskoOmelchenko2019PartII}
\bibitem{KraskoOmelchenko2019PartII}
\bibinfo{author}{E.~\surnamestart Krasko\surnameend} \&
  \bibinfo{author}{A.~\surnamestart Omelchenko\surnameend}
  (\bibinfo{year}{2019}): \emph{\bibinfo{title}{Enumeration of r-regular maps
  on the torus. Part II: Unsensed maps}}.
\newblock {\slshape \bibinfo{journal}{Discrete Math.}}
  \bibinfo{volume}{342}(\bibinfo{number}{2}), pp. \bibinfo{pages}{600--614},
  \doi{10.1016/j.disc.2018.09.004}.

\bibitemdeclare{article}{KraskoOmelchenko2021NoPlanarLoops}
\bibitem{KraskoOmelchenko2021NoPlanarLoops}
\bibinfo{author}{E.~\surnamestart Krasko\surnameend} \&
  \bibinfo{author}{A.~\surnamestart Omelchenko\surnameend}
  (\bibinfo{year}{2021}): \emph{\bibinfo{title}{Enumeration of rooted 4-regular
  maps without planar loops}}.
\newblock {\slshape \bibinfo{journal}{Discrete Math.}}
  \bibinfo{volume}{344}(\bibinfo{number}{8}), p. \bibinfo{pages}{112442},
  \doi{10.1016/j.disc.2021.112442}.

\bibitemdeclare{article}{Kuperberg2003}
\bibitem{Kuperberg2003}
\bibinfo{author}{G.~\surnamestart Kuperberg\surnameend} (\bibinfo{year}{1999}):
  \emph{\bibinfo{title}{What is a virtual link?}}
\newblock {\slshape \bibinfo{journal}{Algebr. Geom. Topol.}}
  \bibinfo{volume}{3}, pp. \bibinfo{pages}{587--591},
  \doi{10.2140/agt.2003.3.587}.

\bibitemdeclare{article}{OmelchenkoTorusMaps}
\bibitem{OmelchenkoTorusMaps}
\bibinfo{author}{A.~\surnamestart Omelchenko\surnameend}
  (\bibinfo{year}{2026}): \emph{\bibinfo{title}{Torus\_Maps: Reference
  implementation and datasets for genus-one projection and diagram enumeration
  on the torus}}.
\newblock \urlprefix\url{https://github.com/avo-travel/Torus_Maps}.

\bibitemdeclare{article}{Omelchenko2007}
\bibitem{Omelchenko2007}
\bibinfo{author}{V.~R.~Meshkov \surnamestart S.~A.~Grishanov\surnameend} \&
  \bibinfo{author}{A.~V. \surnamestart Omelchenko\surnameend}
  (\bibinfo{year}{2007}): \emph{\bibinfo{title}{Kauffman-type polynomial
  invariants for doubly periodic structures}}.
\newblock {\slshape \bibinfo{journal}{Journal of Knot Theory and Its
  Ramifications}} \bibinfo{volume}{16}(\bibinfo{number}{06}), pp.
  \bibinfo{pages}{779--788}, \doi{10.1142/S021821650700549X}.

\bibitemdeclare{article}{Grishanov2009a}
\bibitem{Grishanov2009a}
\bibinfo{author}{V.~Meshkov \surnamestart S.~Grishanov\surnameend} \&
  \bibinfo{author}{A.~\surnamestart Omelchenko\surnameend}
  (\bibinfo{year}{2009}): \emph{\bibinfo{title}{A Topological Study of Textile
  Structures. Part I: An Introduction to Topological Methods}}.
\newblock {\slshape \bibinfo{journal}{Textile Res. J.}} \bibinfo{volume}{79},
  pp. \bibinfo{pages}{702 -- 713}, \doi{10.1177/0040517508095600}.

\bibitemdeclare{article}{Grishanov2009b}
\bibitem{Grishanov2009b}
\bibinfo{author}{V.~Meshkov \surnamestart S.~Grishanov\surnameend} \&
  \bibinfo{author}{A.~\surnamestart Omelchenko\surnameend}
  (\bibinfo{year}{2009}): \emph{\bibinfo{title}{A Topological Study of Textile
  Structures. Part II: Topological Invariants in Application to Textile
  Structures}}.
\newblock {\slshape \bibinfo{journal}{Textile Res. J.}} \bibinfo{volume}{79},
  pp. \bibinfo{pages}{822--836}, \doi{10.1177/0040517508096221}.

\bibitemdeclare{book}{LandoZvonkin2004}
\bibitem{LandoZvonkin2004}
\bibinfo{author}{Sergei \surnamestart S.~K.~Lando\surnameend} \&
  \bibinfo{author}{A.~K. \surnamestart Zvonkin\surnameend}
  (\bibinfo{year}{2004}): \emph{\bibinfo{title}{Graphs on Surfaces and Their
  Applications}}.
\newblock \bibinfo{publisher}{Springer Berlin Heidelberg},
  \bibinfo{address}{Berlin, Heidelberg}, \doi{10.1007/978-3-540-38361-1_1}.

\bibitemdeclare{article}{Tait1877}
\bibitem{Tait1877}
\bibinfo{author}{P.~G. \surnamestart Tait\surnameend} (\bibinfo{year}{1877}):
  \emph{\bibinfo{title}{On knots I, II, III}}.
\newblock {\slshape \bibinfo{journal}{Trans. Roy. Soc. Edinburgh}}
  \bibinfo{volume}{28}, pp. \bibinfo{pages}{35--79}.

\end{thebibliography}

\end{document}